\documentclass[12pt, a4paper]{article}
\usepackage[utf8x]{inputenc}

\usepackage{amsfonts}
\usepackage{multicol}

\usepackage{mathrsfs}
\usepackage{wrapfig}

\usepackage{mathtools}
\usepackage[makeroom]{cancel}
\usepackage[english]{babel}
\usepackage{epsfig}
\usepackage{amssymb, graphicx, amsmath, amsthm}
\usepackage{authblk}
\RequirePackage{graphicx}
\RequirePackage{hyperref}

\usepackage[nottoc]{tocbibind}

\newcounter{dummy} 

\theoremstyle{plain}
\newtheorem{thm}[dummy]{Theorem}
\newtheorem{lem}[dummy]{Lemma}
\newtheorem{dusl}[dummy]{Corollary}
\newtheorem{obs}[dummy]{Observation}

\DeclareMathOperator{\lcm}{\textup{lcm}}
\DeclareMathOperator{\N}{\mathbb{N}} 
\DeclareMathOperator{\Z}{\mathbb{Z}} 
\DeclareMathOperator{\R}{\mathbb{R}} 
\newcommand{\ord}[1]{\textup{ord}_{\mathbb{Z}_{#1}}}
\newcommand{\ordd}[2]{\textup{ord}_{\mathbb{Z}_{#1} \times \mathbb{Z}_{#2}}}
\newcommand{\T}[2]{T_{{#1} \times {#2}}}
\newcommand\myeq[1]{\stackrel{\mathclap{\mbox{#1}}}{=}}
\newcommand*{\email}[1]{%
    \normalsize\href{mailto:#1}{\texttt{#1}}\par
    }

\title{No-three-in-line problem on a torus: periodicity}
\author{Michael Skotnica\thanks{Department of Applied Mathematics, Charles University in Prague, Malostransk\'e n\'am\v{e}st\'i 25, 118 00, Praha 1.}}
\affil{\email{skotnica@kam.mff.cuni.cz}}
\date{}
\begin{document}
\maketitle
\begin{abstract}
	Let $\tau_{m,n}$ denote the maximal number of points on the discrete torus (discrete toric grid)
	of sizes $m \times n$ with no three collinear points. The value $\tau_{m,n}$ is known for the case where 
	$\gcd(m,n)$ is prime. It is also known that $\tau_{m,n} \leq 2\gcd(m,n)$.

	In this paper we generalize some of the known tools for determining $\tau_{m,n}$ and also show some new.
	Using these tools we prove that the sequence $(\tau_{z,n})_{n \in \mathbb{N}}$ is
	periodic for all fixed $z > 1$. In general, we do not know the period; however, if $z =
	p^a$ for $p$ prime, then we can bound it. We prove that $\tau_{p^a,p^{(a-1)p+2}} = 2p^a$
	which implies that the period for the sequence is $p^b$, where $b$ is at most $(a-1)p+2$.
\end{abstract}

\section{Introduction}
	In 1917, Dudeney introduced the \textbf{No-three-in-line-problem} (see problem 317 in \cite{Dudeney}). 
	The goal of this problem is to place as many points as possible on an $n \times n$ grid so that no three points are collinear.
	This problem is still not solved for all $n \in \N$.
	
	In 2012, Fowler, Groot, Pandya and Snapp introduced a modification of this problem:
	\textbf{No-three-in-line-problem on a torus} (see \cite{Fowler}). The purpose is to place as many points as possible on a discrete
	torus of size $m \times n$, where $m,n \in \N$ so that no three of these points are collinear. 
	
	We consider the \emph{discrete torus} as a Cartesian product $\{0, \ldots, m-1\} \times \{0, \ldots, n-1\}  \subset \Z^2$ and we denote it $\T{m}{n}$. 
	A \emph{line} on this torus is an image of a line $\ell$ in $\Z^2$
	under a mapping $\pi_{m,n}$ which maps a point $(x,y) \in \Z^2$ to the 	point $(x \mod{m}, y \mod{n})$.\footnote{
	As a line in $\Z^2$ we consider $\ell_{\R^2} \cap \Z^2$, where $\ell_{\R^2}$ is a line in $\R^2$ which
	contains at least two points of integer coordinates.}
	
	We are interested in estimating the quantity $\tau_{m,n}$ which denotes the maximal number of points that can be placed on the torus $\T{m}{n}$ so that
	no three of these points are collinear.\footnote{From another point of view, we may identify  $\T{m}{n}$ with the group 
	$\Z_m \times \Z_n$. Then lines are cosets of cyclic subgroups $\Z_m \times \Z_n$ generated by $(x,y)$ where $x,y$ are relatively
	prime.}
	
	Fowler et al. in \cite{Fowler} proved that $\tau_{p,p} = p + 1$ and $\tau_{p,p^2} = 2p$ where $p$ is an odd prime. They also showed $\tau_{m,n} = 2$
	whenever $\gcd(m,n) = 1$, where $\gcd$ denotes the greatest common divisor. This fact follows directly from the Chinese remainder theorem.
	These results were generalized by Misiak, St\c{e}pie\'{n}, A. Szymaszkiewicz, L. Szymaszkiewicz and Zwierzchowski  \cite{Misiak} who determined $\tau_{m,n}$ 
	whenever $\gcd(m,n) = p$ is a prime. More concretely, they have shown in this case that $\tau_{m,n} = 2p$ if $m$ or $n$ is divisible by $p^2$ or if $p = 2$ and
	that $\tau_{m,n} = p + 1$ if neither $m$ nor $n$ is divisible by $p^2$ and $p$ is an odd prime; see Theorem 1.2. in \cite{Misiak}.
	
	Misiak et al. also provided a useful general upper bound for $\tau_{m,n}$.
	
	\begin{thm}[Theorem 1.1 in \cite{Misiak}]\label{general_upper}
		Let $m,n \in \N$. Then $\tau_{m,n} \leq 2 \gcd(m,n)$.
	\end{thm}	
	\paragraph{Results of the paper: Periodicity.} We contribute to the study of $\tau_{m,n}$ by showing that once we fix one of the coordinates, then we get a periodic sequence. Namely,
	for a positive integer $z$, we define the sequence $\sigma_z = (\sigma_z(n))_{n = 1}^\infty$ by setting $\sigma_z(n) := \tau_{z,n}$.
	Note that $\sigma_z(n) \leq 2n$ for arbitrary $n$ by Theorem~\ref{general_upper}. In particular, $\sigma_z$ attains only finitely many values.	
	We get following:
	
	\begin{thm}\label{general_period}
		The sequence $\sigma_z$ is periodic for all positive integers $z$ greater than 1.
	\end{thm}
	
	Theorem~\ref{general_period} in principle allows to determine the values $\tau_{z,n}$ for arbitrary big $n$ from several initial values of $n$ (the number of initial values for $n$
	depends on $z$, of course). Unfortunately, for a general $z$ we do not know what is the period. Our proof is purely existential. However, if $z$ is a power
	of a prime, we can say more. We show that in such case the sequence reaches its potential maximum $2z = 2p^a$ for $a \in \N$ and a prime $p$.
	
	\begin{thm}\label{max_p}
		Let $\T{p^a}{p^{(a-1)p+2}}$ be a torus where $p$ is a prime and $a \in \N$. Then $\tau_{p^a,p^{(a-1)p+2}} = 2p^a$.
	\end{thm}		
	
	We can also determine the period of $\sigma_{p^a}$.
	
	\begin{thm}\label{prime_period}
		Let $p$ be a prime, $a \in \N$. Let us denote $m :=\min\{x; \sigma_{p^a}(x) = 2p^a\}$. Then $m = p^b$ for some $b \geq a$ and the sequence $\sigma_{p^a}$ is
		periodic with the~period $m$.
	\end{thm}
	Note that Theorem~\ref{max_p} implies that such an $m$ exists. In fact, we show that each $x$ such that $\sigma_{p^a}(x) = 2p^a$ is a~period
	for $\sigma_{p^a}$. Therefore, from these two theorems above we get that $\sigma_{p^a}$ is periodic with the (not necessarily least) period $p^{(a-1)p+2}$.
	
	\paragraph{Tools.} A useful tool for determining $\tau_{m,n}$ is to find out for which values $m,n,x,y \in \N$ we have $\tau_{m,n} = \tau_{xm,yn}$. Misiak
	et al. showed  $\tau_{p,p} = \tau_{xp,yp}$  if $p$ is odd prime, $x, y$ are relatively prime and neither $x$ nor $y$ is divisible by $p$. In other
	words $\tau_{p,p} = \tau_{a,b}$ if $\gcd(a,b) = p$ and neither $a$ nor $b$ is divisible by $p^2$; see Theorem 4.5 in \cite{Misiak}.  
	
	We generalize this result to the following one.
	\begin{thm}\label{Thm_prev}
		Let $m, n, x, y$ be positive integers such that $m,n$ are not both 1 and $\gcd(x,y) = \gcd(m,y) = \gcd(n,x) = 1$.
		Then $\tau_{xm,yn} = \tau_{m,n}$. 
	\end{thm}
	
	Misiak et al. used the following idea in \cite{Misiak}.  Let us consider tori $\T{xm}{yn}$ and $\T{m}{n}$, 
	where $m,n,x,y \in \N$, a~mapping $f: \T{xm}{yn} \rightarrow \T{m}{n}$ defined by
	$f((a_1,a_2)) = (a_1 \mod{m}, a_2 \mod{n})$, and the~mappings $\pi_{xm, yn}$, $\pi_{m,n}$ defined above from $\Z^2$ to
	$\T{xm}{yn}$, $\T{m}{n}$, respectively. Then $\pi_{m,n} = f \circ \pi_{xm,yn}$. Hence the image of every line on $\T{xm}{yn}$
	under the~mapping $f$ is a line on $\T{m}{n}$.
	Consequently, a set of points on $\T{m}{n}$ such that no three are collinear can be also used on $\T{xm}{yn}$. Indeed,
	if such points were collinear on $\T{xm}{yn}$ they would be collinear also on $\T{m}{n}$ since the image of every line
	on $\T{xm}{yn}$ is a~line on $\T{m}{n}$.  This leads to the following result.	
		
	\begin{lem} [Lemma 4.1(1) in \cite{Misiak}]  \label{thm_geq}
		Let $m,n,x,y$ be positive integers. Then $\tau_{xm,yn} \geq \tau_{m,n}$.
	\end{lem}
	
	Theorem~\ref{Thm_prev} and Lemma~\ref{thm_geq} are our main tools for proving Theorem~\ref{general_period} and Theorem~\ref{prime_period}.
	
	\paragraph{Small values of the sequence $\sigma_z(n)$.} 
	For sake of example, we present here few initial values of $\tau_{z,n}$, our
	table is analogous to the table in \cite{Fowler}. See Table~\ref{table:initial_values}. 
	\begin{table}
	\centering
	\begin{tabular}{l || c | c | c | c | c | c | c | c | c | c | c | c | c | c | c | c | c | c | c | c}
		$z$ & 1 & 2 & 3 & 4 & 5 & 6 & 7 & 8 & 9 & 10 & 11 & 12 & 13 & 14 & 15 & 16 & 17 & 18 & 19 & 20 \ldots\\
		\hline \hline
		2  & 2 & 4 & 2 & 4 & 2 & 4 & 2 & 4 & 2 & 4 & 2 & 4 & 2 & 4 & 2 & 4 & 2 & 4 & 2 & 4 \ldots\\ 
		3  & 2 & 2 & 4 & 2 & 2 & 4 & 2 & 2 & 6 & 2 & 2 & 4 & 2 & 2 & 4 & 2 & 2 & 6 & 2 & 2 \ldots\\
		4  & 2 & 4 & 2 & 6 & 2 & 4 & 2 & 8 & 2 & 4 & 2 & 6 & 2 & 4 & 2 & 8 & 2 & 4 & 2 & 6 \ldots\\
		5  & 2 & 2 & 2 & 2 & 6 & 2 & 2 & 2 & 2 & 6 & 2 & 2 & 2 & 2 & 6 & 2 & 2 & 2 & 2 & 6 \ldots\\
		6  & 2 & 4 & 4 & 4 & 2 & 8 & 2 & 4 & 6 & 4 & 2 & 8 & 2 & 4 & 4 & 4 & 2 & 10 & 2 & 4 \ldots
	\end{tabular}	
	\caption{Initial values of $\tau_{z,n}$.} \label{table:initial_values} 
	\end{table}
	The values $\sigma_p(n)$ for a prime number~$p$ can be fully determined
	from Theorem 1.2 in \cite{Misiak}. In general, it follows from their
	result that the~sequence $\sigma_p(n)$ is periodic with the~period $p^2$ for a
	prime number~$p$.	
	
	Considering $\sigma_4(n)$, we can determine the values $\sigma_4(n)$ again by Theorem~1.2 in \cite{Misiak} for $n \in	\{1,2,3,5,6,7\}$.
	We also have a computer assisted proof that $\sigma_4(4) = 6$.\footnote{\cite{Fowler} also claim this value without a detailed proof;
	with a moderate effort it is also possible to get a computer-free proof.} 
	In particular, a configuration of points showing $\sigma_4(4) \geq 6$ is shown in Figure~\ref{Figure_44_48}.
	We can also easily deduce that $\sigma_4(8) = 8$: due to Theorem~\ref{general_upper}, it
	is sufficient to show that $\sigma_4(8) \geq 8$ which follows from the example in Figure~\ref{Figure_44_48}. Given that $\sigma_4(n) = 8$,
	Theorem~\ref{prime_period} implies that $\sigma_4(n)$ is periodic with the~period 8. 
	\begin{figure}
		\centering
		\includegraphics[scale=1]{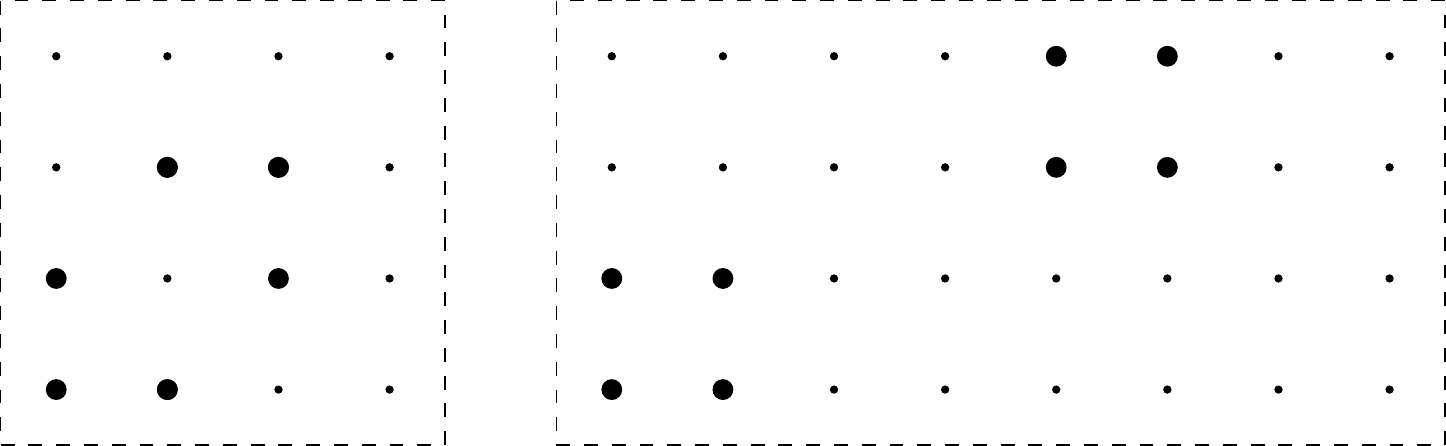}
		\caption{Example of six points on $\T{4}{4}$ (left) and eight points on $\T{4}{8}$ (right) without three points on a~line.}
		\label{Figure_44_48}		
	\end{figure}
		
	\begin{table}
	\centering
	\begin{tabular}{|r || c | c | c |}
		\hline
		$\gcd(n, 2^2 \cdot 3^3)$ & $2^0$ & $2^1$ & $2^2$\\
		\hline \hline
		$3^0$  & 2 & 4 & 4\\
		$3^1$  & 4 & 8 & 8\\
		$3^2$  & 6 & 10 & 12\\
		$3^3$  & 6 & 12 & 12\\
		\hline
	\end{tabular}	
	\caption{Values of $\sigma_6(n)$ according to $\gcd(n, 2^2 \cdot 3^3) = 2^i \cdot 3^j$. } \label{table:sequence_for_sigma_6_GCD} 
	\end{table}	
		
	Regarding the values presented in the table (Table~\ref{table:initial_values}), we can determine exactly $\sigma_6(n)$ for $n \not \equiv 0 \pmod 6$ due to
	Misiak et al. (see \cite{Misiak}). For other values we rely on the~computer proof. It also gives that $\sigma_6(36)=\sigma_6(54)=12$, which is the maximum
	of $\sigma_6(n)$ by Theorem~\ref{general_upper}, and $\sigma_6(24) = 8$.
	
	Unfortunately, we are not able to determine the least period for $\sigma_6(n)$. However, we conjecture $\sigma_6(2^k\cdot 3) = 8$
	for all $k \in \N$.\footnote{We are able to prove it is true for $k \leq 5$ using computer.}
	If it is true then we can determine all values of $\sigma_6(n)$ and the least period is $2^2\cdot 3^3=108$; 
	See Table~\ref{table:sequence_for_sigma_6_GCD}.
	Indeed, by Theorem~\ref{Thm_prev}, Lemma~\ref{thm_geq} and by known initial values we get $\sigma_6(n)=2$ whenever $\gcd(n, 2^2\cdot 3^3) = 1$,
	$\sigma_6(n) = 4$ whenever $\gcd(n, 2^2\cdot 3^3) = 3^1, 2^1,2^2$ or $2^3$,	$\sigma_6(n) = 6$ whenever $\gcd(n, 2^2\cdot 3^3) = 3^2$ or	$3^3$,
	$\sigma_6(n) = 10$ whenever $\gcd(n, 2^2\cdot 3^3) = 2\cdot3^2$, $\sigma_6(n) = 12$	whenever $\gcd(n, 2^2\cdot 3^3) = 2^2\cdot 3^2, 2\cdot3^3$ 
	or $2^2\cdot 3^3$. Finally, $\sigma_6(n) = 8$ whenever $\gcd(n, 2^2\cdot 3^3) = 2\cdot 3, 2^2\cdot 3$ by our conjecture.

\section{Proof of Theorem~\ref{Thm_prev}}
	In this section, we prove Theorem~\ref{Thm_prev}. For this purpose we need the following two well-known theorems.
	
	\begin{thm}[Chinese Remainder Theorem, see \cite{Howard}]  \label{thm_chinese}
		Let $m,n$ be positive integers. Then two simultaneous congruences
		\begin{align*}
			x &\equiv a \pmod{m}, \\
			x &\equiv b \pmod{n} 
		\end{align*}
		are solvable if and only if $a \equiv b \pmod{\gcd(m,n)}$. Moreover, the solution is unique modulo
		$\lcm(m,n)$, where $\lcm$ denotes the least common multiple.
	\end{thm}
	
	\begin{thm}[Dirichlet's Theorem, see \cite{Selberg}] \label{thm_dirichlet}
	Let $a,b$ be positive relatively prime integers. Then there are infinitely many primes of the form $a + nb$, where $n$ is a non-negative
	integer.
	\end{thm}
	
	We again use the idea of Misiak et al.

	\begin{lem}[essentially Lemma 4.1(2) in \cite{Misiak}]\label{Misiak_1}
		Let $\T{xm}{yn}$, $\T{m}{n}$ be tori, where $m,n,x,y \in \N$ and $m,n$ are not both 1. Let
		$f\colon \T{xm}{yn} \to \T{m}{n}$ be a~mapping defined by $f((a_1,a_2)) = (a_1 \mod{m}, a_2 \mod{n})$.
		If the preimage of every line on $\T{m}{n}$ is a line on $\T{xm}{yn}$, then $\tau_{xm,yn} = \tau_{m,n}$. \label{Misiak_1_2}
		\begin{proof}
		
			The Lemma follows from Lemma~\ref{thm_geq} if $\tau_{xm,yn} = 2$. Suppose $\tau_{xm,yn} > 2$.		
			Let $M \subset \T{xm}{yn}$ be a~maximal set such that no three points from $M$ are collinear.
			We show that the image $f(M)$ on $\T{m}{n}$ also satisfies
			the no-three-in-line condition.
			
			First of all, note that $|M| = |f(M)|$. Indeed, if $f(A) = f(B)$ for two distinct $A, B \in M$, then
			$A, B$ and $C$ would be collinear for some point $C \in M \setminus \{A,B\}$.
			
			Now, let us consider three distinct points from $f(M)$. If they were collinear, their preimages would be also collinear
			since the preimage of every line on $\T{m}{n}$ is a line on $\T{xm}{yn}$. 
			Consequently, we have $\tau_{xm,yn} \leq \tau_{m,n}$ and by Lemma~\ref{thm_geq} we have $\tau_{xm,yn} \geq \tau_{m,n}$.
		\end{proof}
	\end{lem}		
	
    Now, we generalize Lemmas 4.3 and 4.4 from \cite{Misiak}.	
	\begin{lem}\label{lemma_1}
		Let $x,y,m,n \in \N$ such that $\gcd(x,y) = \gcd(m,y) = \gcd(n,x) = 1$.
		Let $\T{xm}{yn}$, $\T{m}{n}$ be tori and let $f\colon \T{xm}{yn} \to \T{m}{n}$ be 
		a mapping defined by $f((a_1,a_2)) = (a_1 \mod{m}, a_2 \mod{n})$.	Then the preimage of every line on $\T{m}{n}$ is a line on $\T{xm}{yn}$.
		\begin{proof}
		First, note that it is sufficient to consider only lines which contain the origin.
		The other lines are just translations of such lines.
		
		Let $\ell$ be a line on $\T{m}{n}$ which contains the origin. We can express it as
		$\ell=\{\pi_{m,n}\left(k\left(u,v\right)\right); k \in \Z\}$ where $\gcd(u, v) = 1$.
		We have to find $u^\ast, v^\ast  \in \N_0$ such that $\gcd(u^\ast, v^\ast) = 1$ defining the line 
		$\ell^\ast = \{\pi_{xm,yn}\left(k(u^\ast, v^\ast)\right); k \in \Z\}$ on $\T{xm}{yn}$
		which it is the preimage of $\ell$.
		In other words, for every point $(a_1, a_2$) from $\ell^\ast$ its image $f((a_1, a_2))$ belongs to $\ell$ and for every point $(b_1,b_2)$ from $f^{-1}(\ell)$
		the following equations have solution.
		\begin{align}
			u^\ast k &\equiv b_1 \pmod{xm}, \label{eq_b1} \\
			v^\ast k &\equiv b_2 \pmod{yn}. \label{eq_b2}
		\end{align}		
			
		Let us denote $g_{u,m} := \gcd(u, m)$ and $g_{v,n} := \gcd(v, n)$.
		Since $\frac{u}{g_{u,m}}$ and $\frac{m}{g_{u,m}}$ are relatively prime, by the Dirichlet's Theorem (Theorem~\ref{thm_dirichlet})
		there is $i$ defining a prime $p_u = \frac{u}{g_{u,m}} + i \frac{m}{g_{u,m}}$ which is greater than $mnxy$. 
		At the same way, we can get $j$ defining another prime $p_v = \frac{v}{g_{v,n}} + j \frac{n}{g_{v,n}}$ greater than $mnxy$
		such that $p_v \not=  p_u$. This does not work if either $u$ or $v$ is zero (one of them is
		always non-zero). In such case, we set $i := 1$ or $j := 1$, respectively, and hence $p_u = 1$ or $p_v = 1$, respectively. 		
		Now, we set
		\begin{align}
			u^\ast &:= g_{u,m}p_u = g_{u,m}\left( \frac{u}{g_{u,m}} + i \frac{m}{g_{u,m}}\right) = u + im, \label{p_gamma} \\
			v^\ast &:= g_{v,n}p_v = g_{v,n}\left( \frac{v}{g_{v,n}} + j \frac{n}{g_{v,n}}\right) = v + jn. \label{p_delta}
		\end{align}
		
		Note that $\gcd(g_{u,m}, g_{v,n}) = 1$ and, therefore, $\gcd(u^\ast, v^\ast) = 1$. Indeed, $g_{u,m}$ divides $u$, $g_{v,n}$ divides $v$
		and $\gcd(u, v) = 1$.		
		
		Since $u^\ast = u +im$ and $v^\ast + jn$, the image $f((a_1, a_2))$ belongs to $\ell$ for every $(a_1,a_2)$ from $\ell^\ast$.
		It remains to check whether equalities (\ref{eq_b1}), (\ref{eq_b2}) have solution for every point $(b_1,b_2)$ from $f^{-1}(\ell)$.
		
		Since $(b_1,b_2) \in f^{-1}(\ell)$ there is $t \in \N_0$ such that
		\begin{align*}
			b_1 &\equiv u t \pmod{m}, \\
			b_2 &\equiv v t \pmod{n}.
		\end{align*}
		By the definition of $u^\ast$ and $v^\ast$ (see (\ref{p_gamma}), (\ref{p_delta})) we have
		\begin{align*}
				b_{1} &\equiv u^\ast t \pmod{m},\\
				b_{2} &\equiv v^\ast t \pmod{n},\\
				v^\ast b_{1} &\equiv v^\ast u^\ast t \pmod{v^\ast m},\\ 
				u^\ast b_{2} &\equiv v^\ast u^\ast t \pmod{u^\ast n}.
		\end{align*}
		Therefore, $v^\ast b_{1} \equiv u^\ast b_{2} \pmod{\text{GCD}(v^\ast m, u^\ast n)}$. 
		
		Again, by the definition of $u^\ast$ and $v^\ast$ (see (\ref{p_gamma}), (\ref{p_delta})) we have 
		
		\begin{align*}
		\gcd\left(v^\ast m, u^\ast n \right) &= \gcd\left(g_{v,n}p_{v}m, g_{u,m}p_{u}n\right) 
			= \gcd\left(g_{v,n}p_{v}g_{u,m}\frac{m}{g_{u,m}}, g_{u,m}p_{u}g_{v,n}\frac{n}{g_{v,n}}\right) \\ 
			&=g_{u,m}g_{v,n}\cdot \gcd\left(\frac{m}{g_{u,m}}, \frac{n}{g_{v,n}}\right) = 
			 g_{u,m}g_{v,n} \cdot \gcd\left(x\frac{m}{g_{u,m}}, y\frac{n}{g_{v,n}}\right)  \\
			&= \gcd\left(g_{v,n}p_{v}xg_{u,m}\frac{m}{g_{u,m}}, g_{u,m}p_{u}yg_{v,n}\frac{n}{g_{v,n}}\right) =
			\gcd\left(v^\ast x m, u^\ast y n\right).
		\end{align*}
        Hence $v^\ast b_1 \equiv u^\ast b_2 \pmod{\gcd(v^\ast xm, u^\ast yn)}$.		
		
		Now, let us consider the following equations.
		\begin{align*}
			s &\equiv v^\ast b_1 \pmod{v^\ast xm},   \\
			s &\equiv u^\ast b_2 \pmod{u^\ast yn}.
		\end{align*}		
		Since  $v^\ast b_1 \equiv u^\ast b_2 \pmod{\gcd(v^\ast xm, u^\ast yn)}$, these equations have solution
		by the Chinese Remainder Theorem (Theorem~\ref{thm_chinese}). Since $\gcd(u^\ast, v^\ast) = 1$, $s = u^\ast v^\ast k$ for some $k \in \Z$. 
		Therefore,
		\begin{align*}
			u^\ast k &\equiv  b_1 \pmod{xm}, \\
			v^\ast k &\equiv b_2 \pmod{yn}.
		\end{align*}
		Consequently, equations (\ref{eq_b1}), (\ref{eq_b2}) have solution and the point $(b_1, b_2) \in f^{-1}(\ell)$ lies on $\ell^\ast$. We are done.	
		\end{proof}
	\end{lem}
	Now, we are able to prove Theorem~\ref{Thm_prev} using Lemmas~\ref{Misiak_1} and \ref{lemma_1}.
		\begin{proof}[Proof of Theorem~\ref{Thm_prev}]
			Considering a mapping $f\colon\T{xm}{yn} \to \T{m}{n}$ defined by $f((a_1,a_2)) = (a_1 \mod{m}, a_2 \mod{n})$,
			the preimage of every line on $\T{m}{n}$ is a line on $\T{xm}{yn}$ by Lemma~\ref{lemma_1}. Therefore,
			$\tau_{xm,yn} = \tau_{m,n}$ by Lemma~\ref{Misiak_1}.
		\end{proof}
	
\section{Periodicity}	
	In this section, we prove that the sequence $\sigma_z$ is periodic for all $z \in \N$ (Theorem~\ref{prime_period} and
	Theorem~\ref{general_period}).
	
	Before we start, we mention two lemmas which are used
	in the proofs of Theorem~\ref{prime_period} and \ref{general_period}.
	They are special cases of Theorem~\ref{Thm_prev} and Lemma~\ref{thm_geq}, respectively.
		
	\begin{lem}\label{Lem_prevseq}
		Let $z, x, m \in \N$ such that $z > 1$ and $\gcd(x,z) = 1$. Then $\sigma_z(m) = \sigma_z(xm)$.
		\begin{proof}
			We may use Theorem~\ref{Thm_prev} for the tori $\T{m}{z}$ and $\T{xm}{z}$. 
		\end{proof}
	\end{lem}
	
	\begin{lem}\label{Lem_geqseq}
		Let $z,x,m \in \N$. Then $\sigma_z(m) \leq \sigma_z(xm)$.
		\begin{proof}
			We use Lemma~\ref{thm_geq} for the tori $\T{m}{z}$ and $\T{xm}{z}$. 
		\end{proof}
	\end{lem}
	
	First, we prove Theorem~\ref{prime_period}, that is, the sequence $\sigma_p$ is periodic with the period
	$m :=\min\{x; \sigma_{p^a}(x) = 2p^a\}$ for $p$ prime.
		\begin{proof}[Proof of Theorem~\ref{prime_period}]
			Note that the existence of $m=\min\{x; \sigma_{p^a}(x) = 2p^a\}$ follows from Theorem~\ref{max_p} which is proven in the following section 
			(Section~\ref{sec:proof_thm_max_p}). More precisely, $m \leq p^{(a-1)p+2}$ by Theorem~\ref{max_p}.		
		
			First, we observe that $m = p^b$ for some $b \in \N$. 
			Indeed, if $m$ were $hp^b$ for some  $h > 1$ such that $p$
			does not divide $h$, then $\sigma_{p^a}(p^b) = \sigma_{p^a}(hp^b)$ by Lemma~\ref{Lem_prevseq}. This would contradict the minimality of $m$.
			
			Let us consider an arbitrary $x \in \{1, \ldots, p^b\}$. We show that $\sigma_{p^a}(x) = \sigma_{p^a}(x + \alpha p^b)$ for any $\alpha \in \N$. 
			Since $x \leq p^b$ we can express it as $x = rp^l$, where $0 \leq l \leq b$ and $\gcd(r,p) = 1$.
			By Lemma~\ref{Lem_prevseq} $\sigma_{p^a}(x) = \sigma_{p^a}(p^l)$.
			We consider two cases:
			\begin{enumerate}
				\item $x < p^b$. \newline
					In this case $x + \alpha p^b = p^l(r + \alpha p^{b-l})$. Since $b - l \geq 1$,  we get $\gcd(r + \alpha p^{b-l}, p) = 1$.
					Therefore, $\sigma_{p^a}(p^l(r + \alpha p^{b-l})) = \sigma_{p^a}(p^l) = \sigma_{p^a}(x)$ by Lemma~\ref{Lem_prevseq}.
				\item $x = p^b$. \newline
					In this case $x + \alpha p^b = p^b(1+\alpha)$. Since $\sigma_{p^a}(p^b) = 2p^a = \max \sigma_{p^a}$, 
					Lemma~\ref{Lem_geqseq} implies $\sigma_{p^a}(p^b) = \sigma_{p^a}(hp^b)$ for any $h > 0$ 
					and hence also for $h = (1 + \alpha)$.
			\end{enumerate}
		We proved that $\sigma_{p^a}(x) = \sigma_{p^a}(x + \alpha p^b)$ for any $x \in \{1, \ldots, p^b\}$.
		\end{proof}

	Now, we prove that the sequence $\sigma_{z}(x)$ is periodic for all $z \in N$. For this purpose, we use the following observation and lemma.
	
	\begin{obs}\label{obs:sequence}
		Every infinite sequence $\{C^{(i)}\}_{i \in \N}$ of $n-$tuples of natural numbers $C^{(i)} \in \N^n$ contains an infinite non-decreasing subsequence,
		that is, $\{C^{(t_i)}\}_{i \in \N}$ such that $C^{(t_i)}_j \leq C^{(t_{i+1})}_j$ for all $i,j \in \N$.
		\begin{proof}
			It holds trivially for $n = 1$. For $n > 1$ we may proceed by induction.
		\end{proof}
	\end{obs}
	
	\begin{lem}\label{m_z}
		Let $z \in \N$ and $z = \prod_{i \in I} p_i^{a_i}$ be its prime factorization. There exists $m_z = \prod_{i \in I} p_i^{b}$,
		where $b \geq a_i$ for each $i \in I$ which satisfies the following condition. 
		\begin{align}
			\forall J \subseteq I : \sigma_z\left(\prod_{i \in \overline{J}} p_i^{b} \prod_{i \in J} p_i^{c_i}\right) =
					\sigma_z\left(\prod_{i \in \overline{J}} p_i^{d_i} \prod_{i \in J} p_i^{c_i}\right) \label{cond_mz}
		\end{align}
		for arbitrary $0 \leq c_i < b_i$, $d_i \geq b$ and where $\overline{J} := I \setminus J$.
		\begin{proof}
			First, note that the left expression in (\ref{cond_mz}) is always less than or equal to the right expression
			by Lemma~\ref{Lem_geqseq}.
			The same lemma also implies that if (\ref{cond_mz}) does not hold for some
			$J \subseteq I, (c_i)_{i \in J}, (d_i)_{i \in \overline{J}}$ 
			(that is, the right side of (\ref{cond_mz}) is strictly greater than the left side) then it does not hold also for
			$J \subseteq I, (c_i)_{i \in J}, (\delta_i)_{i \in \overline{J}}$ where all $\delta_i = \max\{d_j; j \in \overline{J}\}$.
			
			For a contradiction, let us assume there is no $m_z$ which satisfies condition~(\ref{cond_mz}). This implies there is an infinite sequence 
			$\mathfrak{C}$ of \emph{counterexamples}
			\begin{align*}
				 \mathcal{C}^{(b)} = \left(J^{(b)}, C^{(b)}, \delta^{(b)}\right), \text{\;where\;} J^{(b)} \subseteq\ I,
				 C^{(b)} = (c^{(b)}_i)_{i \in J^{(b)}}, \delta^{(b)} \geq b,
			\end{align*}	
			for all $b \in \N$ such that 
			\begin{align}
				\sigma_z\left(\prod_{i \in \overline{J^{(b)}}} p_i^{b} \prod_{i \in J^{(b)}} p_i^{c^{(b)}_i}\right) <
					\sigma_z\left(\prod_{i \in \overline{J^{(b)}}} p_i^{\delta^{(b)}} \prod_{i \in J^{(b)}} p_i^{c^{(b)}_i}\right).
					\label{align:sequence_counter}
			\end{align}
			Since $I$ is finite, there exists an infinite subsequence $\mathfrak{C}^\prime$ of counterexamples whose index sets $J^{(b)} \subseteq I$ 
			are equal to some $J \subseteq I$. Since $\sigma_z$ attains only finitely many values, there exists an infinite subsequence
			$\mathfrak{C}^{\prime\prime}$ of counterexamples from $\mathfrak{C}^\prime$
			such that the right sides of (\ref{align:sequence_counter}) for them are equal to some $S \in \N$. Finally, Observation~\ref{obs:sequence} 
			implies there exists an infinite subsequence $\mathfrak{C}^{\prime\prime\prime} = \{\mathcal{C}^{(t_b)}\}_{b \in \N}$ of counterexamples 
			from $\mathfrak{C^{\prime\prime}}$ such that $\{C^{(t_b)}\}_{b \in \N}$ is a non-decreasing sequence.
			
			Now, let $t_b \geq \delta^{(t_1)}$. Then we get
			\begin{align*}
				\sigma_z\left(\prod_{i \in \overline{J}} p_i^{t_b} \prod_{i \in J} p_i^{c^{(t_b)}_i}\right) <
					\sigma_z\left(\prod_{i \in \overline{J}} p_i^{\delta^{(t_b)}} \prod_{i \in J} p_i^{c^{(t_b)}_i}\right) = S &= 
					\sigma_z\left(\prod_{i \in \overline{J}} p_i^{\delta^{(t_1)}} \prod_{i \in J} p_i^{c^{(t_1)}_i}\right) \leq
					\sigma_z\left(\prod_{i \in \overline{J}} p_i^{t_b} \prod_{i \in J} p_i^{c^{(t_b)}_i}\right).
			\end{align*}
			A contradiction. Note that the first inequality holds by~(\ref{align:sequence_counter}). The second inequality holds by Lemma~\ref{Lem_geqseq}
			since $t_b \geq \delta^{(t_1)}$ and $c_i^{(t_1)} \leq c_i^{(t_k)}$ for all $i \in J$.
		\end{proof}
	\end{lem}

	Finally, we prove Theorem~\ref{general_period} using the lemma above.
		\begin{proof}[Proof of Theorem~\ref{general_period}]
			Let $z = \prod_{i \in I} p_i^{a_i}$ be the prime factorization of $z$. We show that $\sigma_{z}$
			is periodic with the period
			$m_z =  \prod_{i \in I} p_i^{b}$ given by Lemma~\ref{m_z} satisfying condition (\ref{cond_mz}).
			
			We aim to show that $\sigma_z(x) = \sigma_z(x + \alpha m_z)$
			for all $x \leq m_z$ and for all $\alpha \in \N_0$.
			We can express $x$ as $x = r \prod_{i \in I} p_i^{c_i}$ such that $\gcd(r, p_i) = 1$ for all $i \in I$. 
			Lemma~\ref{Lem_prevseq} implies $\sigma_z(x) = \sigma_z(\prod_{i \in I} p_i^{c_i})$. We consider two cases:
			\begin{enumerate}
			 \item $c_i < b$ for all $i \in I$. \newline
			 	In this case $x + \alpha m_z = \prod_{i \in I} p_i^{c_i}(r + \alpha \prod_{i \in I} p_i^{b - c_i})$. Therefore,
			 	$\gcd(r + \alpha \prod_{i\in I} p_i^{b - c_i}, p_j) = 1$ for all $j \in I$ since $\gcd(r,p_j) = 1$. Lemma~\ref{Lem_prevseq}
			 	implies	$\sigma_z(x + \alpha m_z) = \sigma_z(\prod_{i \in I} p_i^{c_i}) = \sigma_z(x)$.
			 \item There is  $i \in I$ such that $c_i \geq b$. \newline
				 Let $J := \{i \in I; c_i < b\}$,
				 $K := \{i \in I; c_i = b\}$ and 
				 $L := \{i \in I; c_i > b\}$. Then we get
				 \begin{align*}
				 	x + \alpha m_z = \prod_{i \in J}p_{i}^{c_i}\prod_{i \in K \cup L} p_{i}^{b}
			 		\left(r\prod_{i \in L}p_{i}^{c_i - b} + \alpha \prod_{i \in J}p_{i}^{b - c_i}\right). 
				 \end{align*} 
					 
			 	The expression $\left(r\prod_{i \in L}p_{i}^{c_i - b} + \alpha \prod_{i \in J}p_{i}^{b - c_i}\right)$ is not divisible
			 	by $p_i$ for $i \in J \cup L$. However, it could be divisible by $p_i$ for $i \in K$. Therefore,
			 	\begin{align*}
			 		x + \alpha m_z = h\left(\prod_{i \in J}p_{i}^{c_i}\prod_{i \in K \cup L} p_{i}^{d_{i}}\right)
			 	\end{align*}
			 	for some $d_i \geq b$ where $i \in K \cup L$\footnote{In fact, $d_i = b$ for $i \in L$.} and $h \in \N$ 
			 	such that $\gcd(h, p_i) = 1$ for all $i \in I = J \cup K \cup L$.
			 	Now, we use the property of $m_z$ (see (\ref{cond_mz})) and  Lemma~\ref{Lem_prevseq}:
			 	\begin{align*}
			 	\sigma_z\left(x + \alpha m_z\right) &= 
			 	\sigma_z\left( h\prod_{i \in J}p_{i}^{c_i}\prod_{i \in K \cup L} p_{i}^{d_{i}}\right)
			 		\myeq{$\ast$}\sigma_z\left(\prod_{i \in J}p_{i}^{c_i}\prod_{i \in K \cup L} p_{i}^{d_{i}}\right) \myeq{$\star$} 
			 		 \sigma_z\left(\prod_{i \in J}p_{i}^{c_i}\prod_{i \in K \cup L} p_{i}^{b}\right) \\
                    &\myeq{$\triangle$} \sigma_z\left(\prod_{i \in J}p_{i}^{c_i}\prod_{i \in K \cup L} p_{i}^{c_i}\right)
			 		= \sigma_z\left(\prod_{i \in I}p_{i}^{c_i}\right) = \sigma_z(x),
			 	\end{align*}
			 	where $\ast$ holds by Lemma~\ref{Lem_prevseq}, $\star$ and $\triangle$ hold by
			 	the property of $m_z$ (see (\ref{cond_mz}))
			 	since $c_i < b$ for $i \in J$, $d_i \geq b$ for $i \in K \cup L = \overline{J}$ ($\star$) and 
			 	$c_i \geq b$ for $i \in K \cup L = \overline{J}$ ($\triangle$).
				\end{enumerate}
		 	That is, $\sigma_z$ is periodic for all $z > 1$ with the period $m_z$.
		\end{proof}

	Note, that it is Lemma~\ref{m_z} what makes the proof of this theorem existence.
	
\section{Proof of Theorem~\ref{max_p}}	\label{sec:proof_thm_max_p}
	In the last section, we prove Theorem~\ref{max_p}.
	 
	First of all, we need a tool for determining whether three points are collinear. Let $A = (a_1, a_2),B = (b_1, b_2),C = (c_1, c_2)$ 
	be the points on the torus $\T{m}{n}$.	Let $D(A,B,C)$ denote the following determinant.
	\begin{align*}
		\begin{vmatrix}
		a_1 & b_1 & c_1 \\ 
		a_2 & b_2 & c_2 \\
		1 & 1 & 1 \\ 
		\end{vmatrix}.
	\end{align*}

	\begin{lem}[see Lemma 3.2. in \cite{Misiak}] \label{lem_det}
		 Let $m,n$ be positive integers and let $A = (a_1, a_2),B = (b_1, b_2),C = (c_1, c_2)$ be the points on the torus $\T{m}{n}$.
		 \begin{enumerate}
		 	\item $A,B,C$ are not collinear if and only if there exist $\alpha_1, \alpha_2, \beta_1, \beta_2, \delta_1, \delta_2, \gamma_1, \gamma_2 \in \Z$
			 such that $D((a_1 + \alpha_1m, a_2 + \alpha_2n), (b_1 + \beta_1m, b_2 + \beta_2n), (c_1 + \gamma_1m, c_2 + \gamma_2n)) = 0$.
			\item If $A,B,C$ are collinear then $D((a_1, a_2), (b_1, b_2), (c_1, c_2)) \equiv 0 \pmod{\gcd(m,n)}$.
		 \end{enumerate}
	\end{lem}
	
	In our proof of Theorem~\ref{max_p} we will also need to be able to determine the length of a line on a torus.
	
	\begin{lem}\label{lem_length}
		Let $m,n$ be positive integers and $\ell = \{\pi_{m,n}(A + k (u, v)); k \in \mathbb{Z}\}$ be a line on $\T{m}{n}$,
		where $A, (u,v) \in T_{m \times n}$ such that $ \gcd(u, v) = 1$.
		Then the length of $\ell$ is $\lcm\left(\frac{m}{\gcd(m, u)}, \frac{n}{\gcd(n, v)}\right)$. 
		\begin{proof}
			The line $\ell$ is given by the vector $(u, v)$ and its length is equal to the order of the element $(u, v)$ in
			$\mathbb{Z}_{m} \times \mathbb{Z}_{n}$. Which is
			\begin{align*}
				\ordd{m}{n}(u, v) = 
				\text{lcm}(\text{ord}_{\mathbb{Z}_{m}}(u), \ord{n}(v)).
			\end{align*}		  
	    	The order of an element $h$ in the group $\Z_{x}$ is
			\begin{align*}
				 \ord{x}(h) = \frac{\lcm(x,h)}{h}= \frac{\frac{xh}{ \gcd(x,h)}}{h} = \frac{x}{\gcd(x,h)}.
			\end{align*}	    	    
		\end{proof}
	\end{lem}
	
	We also need lemmas about properties of lines on a torus.	
	
	\begin{lem}\label{lastP_l1}
		Let $m,n$ be positive integers such that $m$ is divisible by $p$ if and only if $n$ is divisible by $p$ for every prime $p$.
		Let $(a_1, a_2)$ be a~point on $\T{m}{n}$ such that $\gcd(a_1,a_2) = 1$. Then there is exactly one line containing
		the origin and $(a_1, a_2)$.
	\begin{proof}
		Let $\ell$ denote the line $\{\pi_{m,n}(k(a_1,a_2)); k \in \Z\}$. This line contains the origin and $(a_1,a_2)$.
		Let us assume that there is another line $\ell' = \{\pi_{m,n}(k(b_1,b_2)); k \in \Z\}$ containing the origin and $(a_1,a_2)$.
		That means there exists $k \in \Z$ such that
		\begin{align*}
			kb_1 &\equiv a_1 \pmod{m}, \\
			kb_2 &\equiv a_2 \pmod{n}.
		\end{align*}
		Therefore, $\ell \subseteq \ell'$. Moreover, $a_1 = kb_1 - s_1 m$ and $a_2 = k b_2 - s_2 n$ for some $s_1, s_2 \in \Z$.
		Since $\gcd\left(m, kb_1 - s_1m\right) = \gcd\left(m,kb_1\right)$ and
		$\gcd\left(n,kb_2 - s_2n\right) = \gcd\left(n,kb_2\right)$, the length of the line $\ell$ is	
		$\lcm\left(\frac{m}{\gcd\left(m,kb_1\right)}, \frac{n}{\gcd\left(n,kb_2\right)}\right)$ by Lemma~\ref{lem_length}.
		By the same lemma the length of $\ell'$ is $\lcm\left(\frac{m}{\gcd\left(m,b_1\right)}, \frac{n}{\gcd\left(n,b_2\right)}\right)$.
		
		Now, we show that $\gcd\left(m,kb_1\right) = \gcd\left(m,b_1\right)$.
		For a contradiction, let us assume $\gcd\left(m,kb_1\right) > \gcd\left(m,b_1\right)$. Then there is a prime $p$ which divides both
		$k$ and $m$.
		By the assumption $p$ also divides $n$ and thus $\gcd(a_1,a_2)= \gcd(kb_1 - s_1 m, k b_2 - s_2 n) \geq p \neq 1$.
		A contradiction. Analogously we get $\gcd\left(n,kb_2\right) = \gcd\left(n,b_2\right)$. Consequently,
		$\ell$ has the same length as $\ell'$. Moreover, $\ell = \ell'$ since $\ell \subseteq \ell'$.
	\end{proof}	
	\end{lem}
	
	\begin{lem}\label{lastP_l2}
		Let $\T{p^a}{p^b}$ be a torus where $p$ is a prime, $a \leq b$, and let $(x,y)$ be a point on 
		$\T{p^a}{p^b}$. Let $g$ denote $\gcd(x,y)$. If $x$ or $y$ is not divisible by $p$,
		then there is exactly one line between the origin and $(x,y)$.
		\begin{proof}
			Let $\ell = \{\pi_{p^a,p^b}(k(v_1, v_2)); k \in \Z\}$ be a line which contains $(x,y)$. Then for some $k \in \Z$
			\begin{align*}
				kv_1 &\equiv x \pmod{p^a}, \\
				kv_2 &\equiv y \pmod{p^b}.
			\end{align*}
			Let $g$ denote $\gcd(x,y)$. Since $\gcd(g,p) = 1$, $g$ has an inverse element modulo $p^b$. Let us denote it $g^{-1}$. Then
			\begin{align*}
				g^{-1}kv_1 &\equiv \frac{x}{g} \pmod{p^a}, \\
				g^{-1}kv_2 &\equiv \frac{y}{g} \pmod{p^b}.
			\end{align*}
			Hence every line on $\T{p^a}{p^b}$ contains $(x,y)$ if and only if it contains $(\frac{x}{g}, \frac{y}{g})$.
			Since $\gcd(\frac{x}{g}, \frac{y}{g}) = 1$, Lemma~\ref{lastP_l1} implies that
			there is exactly one line containing the origin and $(\frac{x}{g}, \frac{y}{g})$. Consequently, there is exactly one line containing
			the~origin and the point $(x,y)$.
		\end{proof}
	\end{lem}
	
	\begin{dusl} \label{cor:lines}
		Let $\T{p^a}{p^b}$ be a torus where $p$ is a prime and let $A = (a_1,a_2), B=(b_1,b_2) \in \T{p^a}{p^b}$ be points on that~torus
		such that $a_1 - b_1$ is not divisible by $p$. Then there is exactly one line between $A,B$ and its length is
		$\max\left\lbrace p^a,\frac{p^b}{\gcd(p^b, a_2 - b_2)}\right\rbrace$.
		\begin{proof}
			The lines between $(a_1,a_2),(b_1,b_2)$ are in one-to-one correspondence with the lines between the origin and 
			$P := \pi_{p^a,p^b}(a_1-b_1, a_2-b_2)$. Since $a_1-b_1$ is not divisible by $p$ there is exactly one line between the origin and $P$
			by Lemma~\ref{lastP_l2}.
			
			Such line is given by the vector $\pi_{p^a,p^b}(\frac{a_1-b_1}{\gcd(a_1-b_1, a_2-b_2)},\frac{a_2-b_2}{\gcd(a_1-b_2, a_2-b_2)})$.
			We can express it as $(v,wp^r)$, where $v,w$ is not divisible by $p$ and $r < b$. Moreover, 
			$p^r = \gcd(p^b, \frac{a_2-b_2}{\gcd(a_1-b_1, a_2-b_2)}) = \gcd(p^b, a_2-b_2)$ since $a_1 - b_1$ is not divisible by $p$.
			Lemma~\ref{lem_length} implies the length of the line between the origin and $P$ is 
			\begin{align*}
				\lcm\left(\frac{p^a}{\gcd(p^a, v)},\frac{p^b}{\gcd(p^b, wp^r)}\right) = \lcm\left(p^a, \frac{p^b}{p^r}\right)
				= \max\left\lbrace p^a,\frac{p^b}{p^r} \right\rbrace =  \max\left\lbrace p^a,
				\frac{p^b}{\gcd(p^b, a_2 - b_2)}\right\rbrace.
			\end{align*}
		\end{proof}
	\end{dusl}	
	
	\begin{lem}\label{lastP_l3}
		Let $\T{p^a}{p^b}$ be a torus where $p$ is a prime and let $A,B, C$
		be points on it such that $\pi_{p^{a-1}, p^{b}}(A) = \pi_{p^{a-1}, p^{b}}(B)$.
		If each line $\ell$ containing
		the points $A, C$ on $\T{p^a}{p^b}$ has the same length as its image $\ell' := \pi_{p^{a-1}, p^{b}}(\ell)$
		then $A,B,C$ are not collinear. (See Figure~\ref{Ex_1} for an example.)
		\begin{proof}
			If there is a line $\ell$ which contains $A,B,C$ then its image $\ell^\prime$
			has smaller length since $\pi_{p^{a-1}, p^{b}}(A) = \pi_{p^{a-1}, p^{b}}(B)$. However, it is impossible by our assumption.
		\end{proof}		
 	\end{lem}
 	
 	\begin{figure}
 		\centering
 		\includegraphics[scale=1]{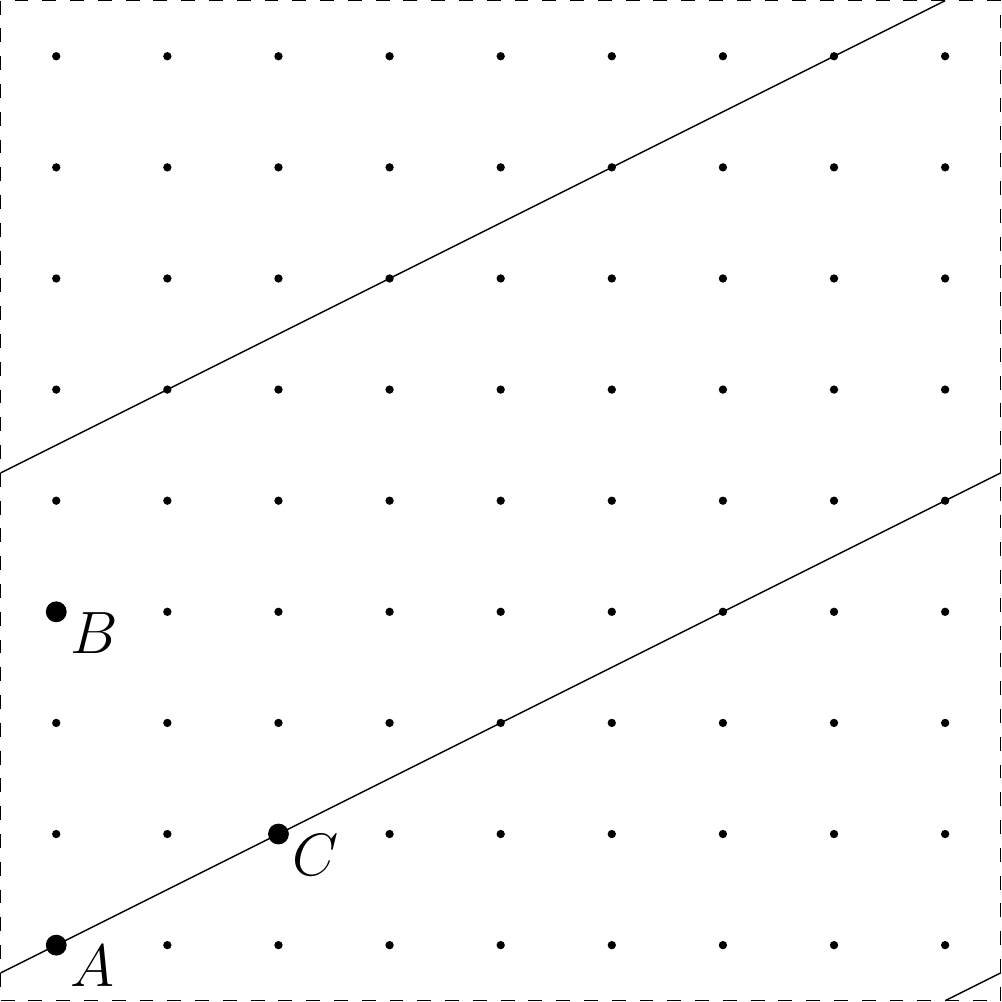}
		\\[2mm]
 		\includegraphics[scale=1]{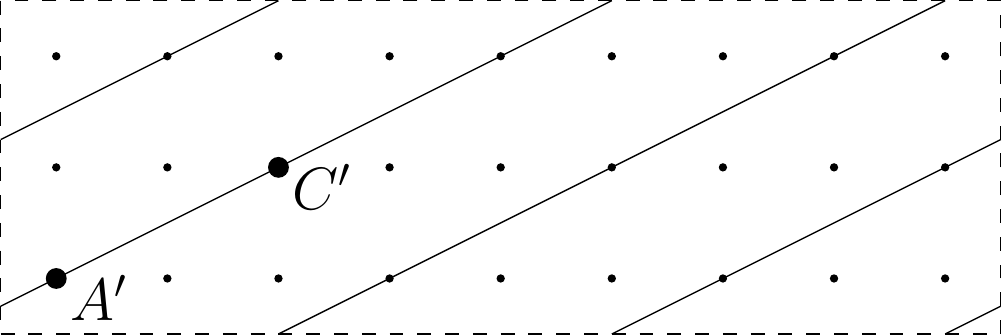}
 		\caption{An application of Lemma~\ref{lastP_l3}. Let $A = (0,0),B =(3,0)$ and $C = (1,2)$ be points on $T_{9,9}$.
 		There is exactly one line between $A$ and $C$ by Lemma~\ref{lastP_l1} and this line has the length $9$ (see the upper picture). 		
 		Now, let us consider points $A^\prime =\pi_{3,9}(A)$ and $C^\prime = \pi_{3,9}(C)$ on $T_{9,9}$. The line containing
 		$A^\prime$ and $C^\prime$ also has the length 9 (see the lower picture). Since $A^\prime = \pi_{3,9}(B)$,
 		we conclude that $A,B,C$ are not collinear.}
 		\label{Ex_1} 		
 	\end{figure}
 	
 	\begin{lem} \label{lem_more_lines}
 		Let $\T{p^a}{p^b}$ be a~torus where $p$ is a prime and $a \leq b$. Let
 		$(v_1p^c,v_2p^d)$ be a~point on $\T{p^a}{p^b}$ such that $p$ does not divide 
 		$v_1$ or $v_2$, $c < a$, $d < b$ and $c \leq d$. Then each line which contains the origin and 
 		$(v_1p^c,v_2p^d)$ also contains some point $(w_1,w_2)$ such that $w_1, w_2$ satisfy the following equations.
 		\begin{align*}
 			w_1 &\equiv v_1 \pmod{p^{a-c}}, \\
 			w_2 &\equiv v_2p^{d-c} \pmod{p^{b-c}}.
 		\end{align*}
 		\begin{proof}
 			Let $\ell$ be a line which contains the origin and $(v_1p^c,v_2p^d)$. Then we can express it as
 			$\ell = \{\pi_{p^a,p^b}(k(u_1,u_2)); k \in \Z\}$ such that $\gcd(u_1,u_2) = 1$ and there exists $k \in \Z$
 			which satisfies the following equations.
 			\begin{align*}
 				ku_1 &\equiv v_1p^c \pmod{p^a}, \\
 				ku_2 &\equiv v_2p^d \pmod{p^b}.
 			\end{align*}
 			Since $\gcd(u_1,u_2) = 1$, $p$ does not divide $u_1$ or $u_2$. Therefore, $p^c$ has to
 			divide $k$ since $c \leq d$. Consequently,
 			\begin{align*}
 				w_1 := \frac{k}{p^c}u_1 &\equiv v_1 \pmod{p^{a-c}}, \\
 				w_2 := \frac{k}{p^c}u_2 &\equiv v_2p^{d-c} \pmod{p^{b-c}}.
 			\end{align*}
 		\end{proof}
 	\end{lem}
 	Now, we are able to prove Theorem~\ref{max_p}.
 		\begin{proof}[Proof of Theorem~\ref{max_p}]
 			Theorem~\ref{general_upper} implies $\tau_{p^a,p^{(a-1)p+2}} \leq 2p^a$. Therefore, it is sufficient
 			to show a~construction of $2p^a$ points on $\T{p^a}{p^{(a-1)p+2}}$ where no three of them are collinear.
 			
 			Let  $X :=\{(i, i^2p); i \in P\}$ and $Y := \{(i, i^2p + 1); i \in P\}$, where $P = \{0, \ldots, p-1\}$. Misiak et al. 
 			(see Theorem 1.2(3a) in \cite{Misiak}) proved
 			that no three points from the set $\pi_{p,p^2}(X \cup Y)$ are collinear on the torus $\T{p}{p^2}$. 
 			In this proof, we inductively construct a set $X_a \cup Y_a \subset \Z^2$ by taking multiple copies 
 			of $X \cup Y$ and prove by induction
 			that no three points from $\pi_{p^a, p^{(a-1)p+2}} (X_a \cup Y_a)$ are collinear on the torus $\T{p^a}{p^{(a-1)p+2}}$.

 			First, let us define the set $X_a \cup Y_a$.
 			\begin{enumerate}
 				\item For $a = 1$ we define $X_a := X$ and $Y_a := Y$.
 				\item For $a > 1$ we take $p$ copies of the set $X_{a-1} \cup Y_{a-1}$ on $\Z^2$ so that 
 				\begin{align*}
 					X_{a} &:= X_{a-1} \cup \bigcup_{i=1}^{p-1} X_{a-1} + \left(ip^{a-1}, p^{(a-2)p + i + 3}\right), \\
	 				Y_{a} &:= Y_{a-1} \cup \bigcup_{i=1}^{p-1} Y_{a-1} + \left(ip^{a-1}, p^{(a-2)p + i + 3}\right).
 				\end{align*}
 			\end{enumerate}		
 			\begin{figure}
				\centering
				\includegraphics[scale=1]{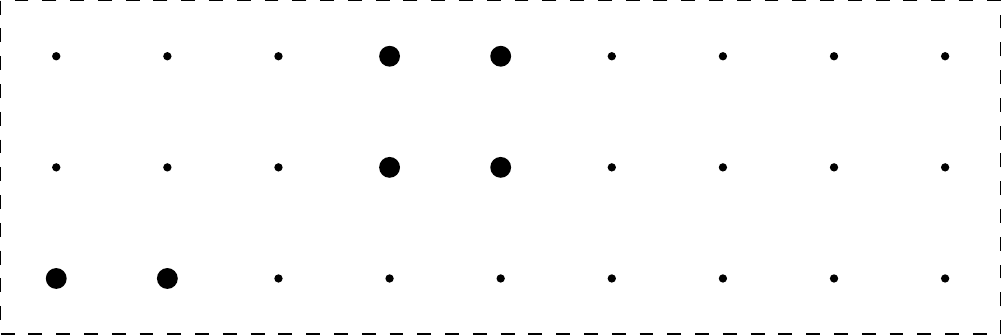}
				\\[4mm]
				\includegraphics[scale=1]{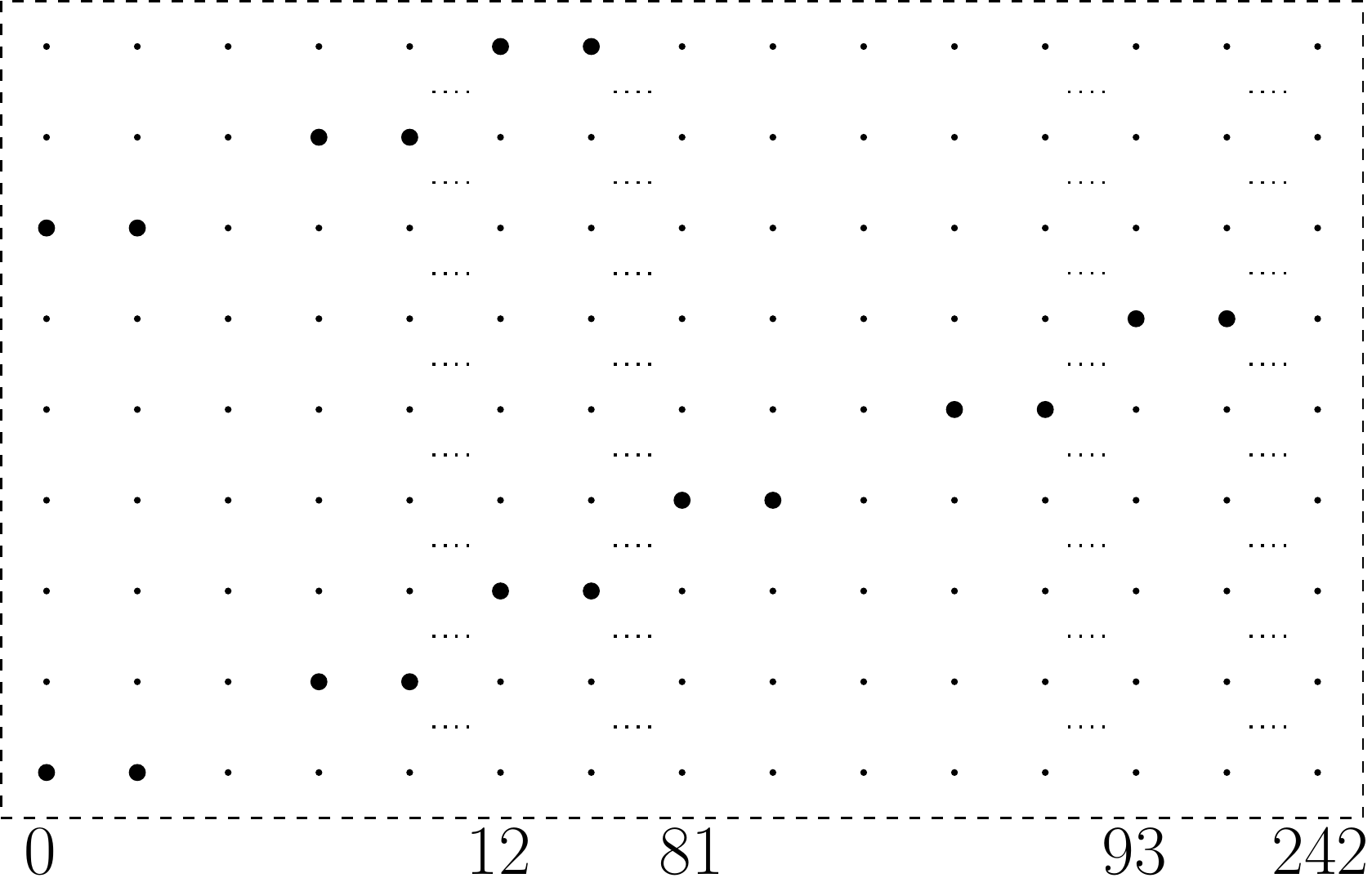}
				\caption{An example of the construction of $\pi_{p^a, p^{(a-1)p+2}} (X_{a} \cup Y_{a})$ from Theorem~\ref{max_p} for $p = 3$, $a = 1$ (see the upper picture) and $a=2$ (see the lower picture).}
				\label{ex1XY}	
			\end{figure}	
			\begin{figure}
				\centering
				\includegraphics[scale=1.2]{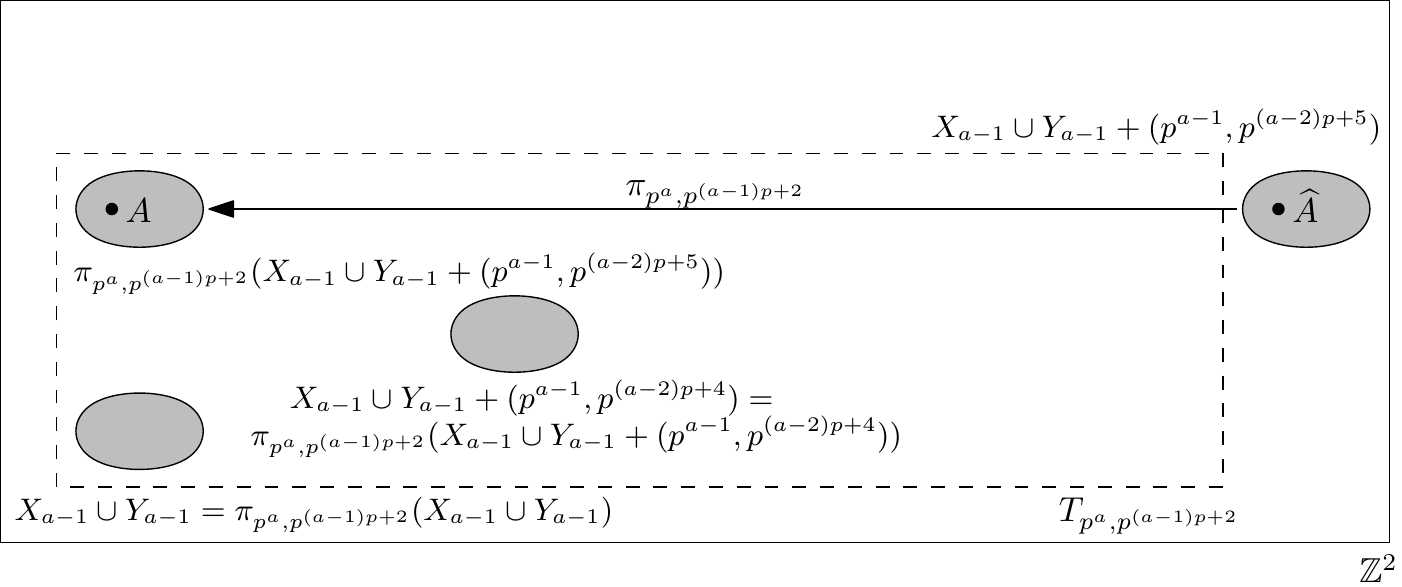}
				\caption{A schematic picture of the construction of $\pi_{p^a, p^{(a-1)p+2}} (X_{a} \cup Y_{a})$
				from Theorem~\ref{max_p} for $p = 3$ and an example of the preimage $\widehat{A}$ of a point $A
				\in \pi_{p^a, p^{(a-1)p+2}} (X_{a} \cup Y_{a})$.} \label{fig:schematicExample}
			\end{figure}
			As an example of the construction see Figures~\ref{ex1XY} and \ref{fig:schematicExample}. 						
			For a point $A \in \pi_{p^a, p^{(a-1)p+2}}(X_a \cup Y_a)$ on $\T{p^a}{p^{(a-1)p+2}}$ we call the point $\widehat{A} \in X_a \cup Y_a$ such that
			$\pi_{p^a, p^{(a-1)p+2}}(\widehat{A}) = A$ the \emph{preimage} of $A$ (see Figure~\ref{fig:schematicExample}).
 			Now, we inductively prove that $\pi_{p^a, p^{(a-1)p+2}}(X_a \cup Y_a)$ does not contain three collinear points.
 			Case $a = 1$ was proved by Misiak et al. (see Theorem 1.2(3a) in \cite{Misiak}).
 			 
 			Let us assume $a > 1$ and
 			let	$A,B,C \in \pi_{p^a,p^{(a-1)p+2}}(X_a \cup Y_a)$. We have to consider several cases. 
      		\smallskip
      		
      		First, let us assume that $\pi_{p^{a-1}, p^{(a-2)p+2}}(A), \pi_{p^{a-1}, p^{(a-2)p+2}}(B), \pi_{p^{a-1}, p^{(a-2)p+2}}(C)
      		\in X_{a-1} \cup Y_{a-1}$ are three distinct points.
      		In other words, the preimages of $A,B,C$ in $X_a \cup Y_a$ are copies of three distinct points from $X_{a-1} \cup Y_{a-1}$.
			If $A,B,C$ are collinear on $\T{p^a}{p^{(a-1)p+2}}$, then $\pi_{p^{a-1}, p^{(a-2)p+2}}(A)$, $\pi_{p^{a-1}, p^{(a-2)p+2}}(B)$,
			$\pi_{p^{a-1}, p^{(a-2)p+2}}(C)$ are
			collinear on $\T{p^{a-1}}{p^{(a-2)p+2}}$ as well which contradicts the~induction hypothesis.
			 			
 			Now, we assume $A$, $B$, $C$ satisfy $\pi_{p^{a-1},p^{(a-2)p+2}}(A) = \pi_{p^{a-1}, p^{(a-2)p+2}}(B)$. In other words,
 			the preimages of $A$ and $B$ are copies of the same point from $X_{a-1} \cup Y_{a-1}$.
 			\bigskip
 			
 			In the first case, we assume that the preimages of $A,B$ are from $X_a$ and the preimage of $C$ is from $Y_a$.
 			Therefore,
 			\begin{align*}
 				A &= (k + t_1p, k^2p + t_2p), \\
 				B &= (k + t_1p + up^{a - 1}, k^2p + t_2p+ xp^q), \\
 				C &= (m +v_1p, m^2p + v_2p + 1)
 			\end{align*}
 			for some $k,m \in \{0, \ldots, p-1\}$, $u \in \{1, \ldots, p-1\}$, 
 			for suitable $t_1, t_2, v_1, v_2, x \in \N_0$, $q \geq (a-2)p + 4 \geq a$ and where $p$ does not divide $x$.
 			Note that $u > 0$; otherwise $A=B$.
 			We may use the~translation given by the vector $(-t_1p, -t_2p)$ and we get
 			\begin{align*}
 				A_t &= (k, k^2p), \\
 				B_t &= (k + up^{a - 1}, k^2p + xp^q), \\
 				C_t &= (m + s_1p, m^2p + s_2p + 1),
 			\end{align*}
 			where $s_1 = v_1 - t_1 \mod p^a$ and $s_2 = v_2 - t_2 \mod p^{(a-1)p+2}$. The points $A_t, B_t, C_t$ are collinear if and only if
 			$A,B,C$ are collinear. 
 			We compute $D(A_t,B_t,C_t)$.
 			\begin{align*}
 				D(A_t,B_t,C_t) &=D((k, k^2p), (k + up^{a - 1}, k^2p +xp^q), (m +ps_1, m^2p + ps_2 + 1)) \\
 				&= kxp^q+m^2up^a+s_2up^a+up^{a-1}-mxp^q-s_1xp^{q+1} - k^2up^a \\
 				&= p^{a-1}(u + pR)
 			\end{align*}
 			
 			for some $R \in \Z$. Therefore, $D(A_t,B_t,C_t) \not\equiv 0 \pmod{p^a}$ since $u >0$.
 			Consequently, $A,B,C$ are not collinear on $\T{p^a}{p^{(a-1)p+2}}$ by Lemma~\ref{lem_det} (2).
 			
 			Similarly, we check the case when the preimages of $A,B$ are from $Y_a$ and the preimage of $C$ from $X_a$.
 			Hence
 			\begin{align*}
 				A &= (k + t_1p , k^2p + t_2p + 1),\\
				B &= (k + t_1p + up^{a - 1}, k^2p + t_2p + xp^q + 1),\\							
	    	C &= (m + v_1p, m^2p + v_2p).
 			\end{align*}
 			We use the translation given by the vector $(-t_1p, -t_2p)$ as we did in the previous case:
 			\begin{align*}
 				A_t &= (k, k^2p + 1),\\
				B_t &= (k + up^{a - 1}, k^2p + xp^q + 1),\\							
	    	C_t &= (m + s_1p, m^2p + s_2p)
 			\end{align*}
 			and we get
 			\begin{align*}
 			  D(A_t,B_t,C_t) &= kxp^q + m^2up^a + s_2up^a - up^{a-1} - mxp^q - s_1p^{q+1} - k^2up^a \\
 			  &=p^{a-1}(-u + pR).
 			\end{align*}
 			Therefore, $A,B,C$ are not collinear on $\T{p^a}{p^{(a-1)p+2}}$ by the same lemma.
 			
 			\bigskip
 			Now, we check the  case when the preimages of all three points are from $X_a$.
 			The case where they are all from $Y_a$ is just a translation.
	 		We can express the points as follows:
	 		\begin{align*}
				A &= (k + s_1, k^2p + s_2),\\
				B &= (k + s_1 + up^{a-1}, k^2p + s_2 + xp^q),\\
				C &= (m + s_3, m^2p + s_4),
			\end{align*}
			where $p$ divides $s_1, s_3$ and $p^4$ divides $s_2, s_4$ by the definition of the construction (the first iteration
			for $a = 2$), $q\in \{p^{(a-2)p + 4}, \ldots, p^{(a-1)p + 2}\}$, $u \geq 1$ is not divisible by $p$, $x \in \{0,1\}$ and
			$k,m \in \{0, \ldots,  p-1\}$.
			\begin{enumerate}
				\item First, we assume $m \neq k$. Let us map these points to the smaller torus $\T{p^{a}}{p^{(a-2)p + 4}}$
				using $\pi_{p^{a},p^{(a-2)p + 4}}$ and let $A'$, $B'$, $C'$ be the images.
				Then
				\begin{align*}
					A' &= (k + s_1, k^2p + s_2'),\\
				  	B' &= (k + s_1 + up^{a-1}, k^2p + s_2'),\\
				  	C' &=(m + s_3, m^2p + s_4'),
				\end{align*} 	
				where $p^4$ divides $s_2', s_4'$. Let us check the line between $A'$ and $C'$. Since $((k - m) + s_1 - s_3)$ 
				is not divisible by $p$, Corollary~\ref{cor:lines} implies there is exactly one line between them and  its length is
				$d = \max \{p^a,\frac{p^{(a-2)p + 4}}{\gcd(p^{(a-2)p + 4},p(k-m)(k+m) + s_2' - s_4')}\} = \max\{p^a, p^{(a-2)p + 4 - h}\}$, where 
				$h \in \{1,2\}$ (the sum of $k$ and $m$ may be $p$). 
								
				Now, let us map $A', B', C'$ to the torus $\T{p^{a-1}}{p^{(a-2)p + 4}}$ and let $A''$, $B''$, $C''$ be the images.
				We can express them as
				\begin{align*}
					A'' &= (k + s_1', k^2p + s_2'),\\
					B'' &= A'',\\
	 				C'' &= (m + s_3', m^2p + s_4'),
				\end{align*}
				where $p$ divides $s_1'$ a $s_3'$.
				Again, by Corollary~\ref{cor:lines} there is exactly one line between $A''$ and $C''$ and its length is
				$d^\prime = \max\{p^{a-1}, p^{(a-2)p + 4 - h}\}$. Since $(a-2)p + 4 - h \geq a$ for $a > 1$,
				$d' = d$ and $A', B', C'$ are not collinear by Lemma~\ref{lastP_l3}. Therefore, $A,B,C$ 
				are not collinear on $\T{p^a}{p^{(a-1)p+2}}$.

				\item 
				Let us check the case when $k=m$. That means $\pi_{p,p^2}(A) = \pi_{p,p^2}(C)$. 
				
				First, let us assume that 
				$\pi_{p^{a-1},p^{(a-2)p+2}}(A) \neq \pi_{p^{a-1},p^{(a-2)p+2}}(C)$. In other words, the preimages of
				all three points are copies of the same point from $X$ but they are not copies
				of the same point from $X_{a-1}$. 
				We have
				\begin{align*}
					A &= (k + s_1, k^2p + s_2),\\
					B &= (k + s_1 + up^{a-1}, k^2p + s_2 + xp^q),\\
					C &= (k + s_3, k^2p + s_4),
				\end{align*}
				where $p$ divides $s_1, s_3$ and $p^4$ divides $s_2, s_4$
				 by the definition of the~construction (the first iteration
				for $a = 2$), $q\in \{p^{(a-2)p + 4}, \ldots, p^{(a-1)p + 2}\}$, $u \geq 1$ is not divisible by $p$, $x \in \{0,1\}$
				and $k,m \in \{0, \ldots,  p-1\}$. We can see that
				$p^{a-1}$ does not divide $s_1$ and $p^{(a-2)p + 4}$ does not divide $s_2$ or
				$p^{a-1}$ does not divide $s_3$ and $p^{(a-2)p + 4}$ does not divide $s_4$; otherwise
				$\pi_{p^{a-1},p^{(a-2)p+2}}(A) = \pi_{p^{a-1},p^{(a-2)p+2}}(C)$. Now, we use the translation
				given by the vector $(-k - s_1, -k^2p - s_2)$ and we get
				\begin{align*}
					A_t &= (0,0),\\
					B_t &= (up^{a-1}, xp^q),\\
					C_t &= (vp^t, yp^r),
				\end{align*}
				where $u,x,v,y$ are not divisible by $p$, $t \leq a-2$, $r\leq (a-2)p + 2$.
				Let us have a look at the images of these points on the torus $\T{p^{a}}{p^{(a-2)p + 4}}$.
				We get
				\begin{align*} 
					A' &= (0,0),\\
					B' &= (up^{a-1}, 0),\\
					C' &= (vp^t, y'p^r)
				\end{align*}
				for suitable $y'$ which is not divisible by $p$.
				Let us compute the length of the lines between $A'$ a $C'$. By Lemma~\ref{lem_more_lines},
				every such line passing through some point $(w_1, w_2)$ satisfies
				\begin{align*}
					w_1 &\equiv v \pmod{p^{a - t}}, \\
					w_2 &\equiv y'p^{r-t} \pmod{p^{(a-2)p + 2 -t}}.
				\end{align*}
				By Corollary~\ref{cor:lines}, there is exactly one line between the origin and such $(w_1, w_2)$ and its length is	
				$\max\{p^a, \frac{p^{(a-2)p + 4}}{p^{r-t}}\}$.
				Therefore, each line between $A'$ and $C'$ has the length $d = \max\{p^a, \frac{p^{(a-2)p + 4}}{p^{r-t}}\}$.
				
				Let us map $A^\prime, B^\prime, C^\prime$ to the~torus $\T{p^{a-2}}{p^{(a-2)p + 4}}$
				and let $A^{\prime\prime}, B^{\prime\prime}, 
				C^{\prime\prime}$ be the images. We get
				\begin{align*}
					A'' &= A', \\
					B^{\prime\prime} &= A^{\prime\prime} \\
					C'' &=  (v'p^t, y'p^r)
				\end{align*}
				for suitable $v'$ which is not divisible by $p$. Similarly, the length $d'$ of each line between $A''$ and 
				$C''$ is $d' = \max\{p^{a-1}, \frac{p^{(a-2)p + 4}}{p^{r-t}}\}$ by Corollary~\ref{cor:lines}.
				Note that the maximal $r - t$ satisfying the definition of the construction equals
				$\max_{x \in \{1, \ldots, a-2\}}(xp +2 - x) = (a-2)(p-1) + 2$.
				Hence $r-t \leq (a-2)p - a + 4$ and  $(a-2)p + 4 - (r - t) \geq a$. Consequently, $d =d'$ and
				Lemma~\ref{lastP_l3} implies that $A', B', C'$
				are not collinear and, therefore, $A,B,C$ are not collinear on $\T{p^a}{p^{(a-1)p+2}}$.
	
				Finally, we check the case when all three $A,B,C$ points are copies of one point
				from $X_{a-1}$ which is without loss of generality the origin.
				Otherwise, we use a translation similarly to the previous cases.
				Therefore,
				\begin{align*}
					A=(up^{a-1}, p^q), \\
					B = (vp^{a-1}, p^r),\\
					C =(wp^{a-1}, p^s)
				\end{align*}
				for distinct $q,r, s \in \{p^{(a-2)p + 4}, \ldots, p^{(a-1)p + 2}\}$
				and for distinct $u,v,w \in \{0, \ldots, p-1\}$ satisfying the definition of the~construction.
				Let us denote $b :=(a-1)p + 2$ and $D:= D((up^{a-1} + \alpha_1 p^a, p^q +\alpha_2 p^b), 
				(vp^{a-1} + \beta_1 p^a, p^r + \beta_2 p^b), (wp^{a-1} + \gamma_1 p^a, p^s + \gamma_1 p^b))$
				the determinant for $A,B,C$ from Lemma~\ref{lem_det}~(1), where 
				$\alpha_1, \alpha_2, \beta_1, \beta_2, \gamma_1, \gamma_2 \in \Z$. 
				Then
				\begin{align*}
					D &= D((up^{a-1}, p^q),  (vp^{a-1}, p^r), (wp^{a-1}, p^s))
					+ p^{a}D((\alpha_1, p^q), (\beta_1, p^r), (\gamma_1, p^s))\\
					&+p^bD((up^{a-1}, \alpha_2),(vp^{a-1}, \beta_2),(wp^{a-1}, \gamma_2)) 
					+ p^{a+b}D((\alpha_1,\alpha_2), (\beta_1, \beta_2), (\gamma_1, \gamma_2)) \\ 
					&= p^{a+r-1}(u-w) + p^{a+q-1}(w-v) + p^{a+s-1}(v-u) \\
					&+ p^{a}(p^r(\alpha_1 - \gamma_1) +p^q(\gamma_1 - \beta_1) 
					+p^s(\beta_1 - \alpha_1)) \\
					&+ p^{a+b-1}Q + p^{a+b}W,
				\end{align*}
				where $Q = D((u, \alpha_2),(v, \beta_2),(w, \gamma_2))$ and
				$W= D((\alpha_1,\alpha_2), (\beta_1, \beta_2), (\gamma_1, \gamma_2))$.
				Without loss of generality let $q$ be the smallest number of $\{q,r,s\}$.
				Then indeed $q < b$ and thus we get $D =p^{a+q-1}((w-v) + pH)$ for suitable $H \in \Z$.
				Therefore, $D \neq 0$ and $A,B,C$ are not collinear on $\T{p^a}{p^{(a-1)p+2}}$
				by Lemma~\ref{lem_det}~(1) and we are done.
			\end{enumerate}
	 	\end{proof}		
 	
\section*{Acknowledgement}
	I would like to thank my supervisor Martin Tancer for all his support and advice. I would also like
	to thank the anonymous reviewers for valuable advice and remarks.

\bibliographystyle{alpha} 

\begin{thebibliography}{MSS{\etalchar{+}}16}

\bibitem[Dud17]{Dudeney}
H.~E. Dudeney.
\newblock {\em Amusements in mathematics}.
\newblock Nelson, Edinburgh, 1917.
\newblock pp. 94, 222.

\bibitem[FGPS12]{Fowler}
J.~Fowler, A.~Groot, D.~Pandya, and B.~Snapp.
\newblock The no-three-in-line problem on a torus.
\newblock 2012.
\newblock {arXiv:} 1203.6604.

\bibitem[How02]{Howard}
F.~T. Howard.
\newblock A generalized {C}hinese remainder theorem.
\newblock {\em The College Mathematics Journal}, 33(4):279--282, 2002.

\bibitem[MSS{\etalchar{+}}16]{Misiak}
A.~Misiak, Z.~St\c{e}pie\'{n}, A.~Szymaszkiewicz, L.~Szymaszkiewicz, and
  M.~Zwierzchowski.
\newblock A note on the no-three-in-line problem on a torus.
\newblock {\em Discrete Mathematics}, 339(1):217--221, 2016.

\bibitem[Sel49]{Selberg}
A.~Selberg.
\newblock An elementary proof of {D}irichlet's theorem about primes in an
  arithmetic progression.
\newblock {\em Annals of Mathematics}, 50(2):297--304, 1949.

\end{thebibliography}
\renewcommand{\bibname}{References}
\newcommand{\etalchar}[1]{$^{#1}$}

\end{document}